\documentclass[12pt,draftclsnofoot,journal,onecolumn]{IEEEtran}

\usepackage{cite}
\usepackage{graphicx}
\usepackage{amsmath}
\usepackage{amsfonts}
\usepackage{array}
\usepackage{booktabs}
\usepackage{amssymb}
\usepackage{subfigure}
\newtheorem{lemma}{\textbf{Lemma}}
\newtheorem{proposition}{\textbf{Proposition}}

\newtheorem{remark}{\textbf{Remark}}

\def\bc{{\mathbf{c}}}

\def\bee{{\mathbf{e}}}
\def\bff{{\mathbf{f}}}
\def\bg{{\mathbf{g}}}
\def\bh{{\mathbf{h}}}

\def\bn{{\mathbf{n}}}

\def\bs{{\mathbf{s}}}

\def\bw{{\mathbf{w}}}

\def\bz{{\mathbf{z}}}

\def\bI{{\mathbf{I}}}

\begin{document}

\title{Spatial Intercell Interference Cancellation with CSI Training and Feedback}
\author{Jun Zhang, Jeffrey G. Andrews, and Khaled B. Letaief
\thanks{J. Zhang and K. B. Letaief are with Dept. of ECE, The Hong Kong University of Science and Technology (email: \{eejzhang,eekhaled\}@ust.hk). J. G. Andrews is with Dept. of ECE, The University of Texas at Austin (email: jandrews@ece.utexas.edu).}
}

\maketitle

\begin{abstract}
We investigate intercell interference cancellation (ICIC) with a practical downlink training and uplink channel state information (CSI) feedback model.  The average downlink throughput for such a 2-cell network is derived.  The user location has a strong effect on the signal-to-interference ratio (SIR) and the channel estimation error. This motivates adaptively switching between traditional (single-cell) beamforming and ICIC at low signal-to-noise ratio (SNR) where ICIC is preferred only with low SIR and accurate channel estimation, and the use of ICIC with optimized training and feedback at high SNR. For a given channel coherence time and fixed training and feedback overheads, we develop optimal data vs. pilot power allocation for CSI training as well as optimal feedback resource allocation to feed back CSI of different channels.  Both analog and finite-rate digital feedback are considered. With analog feedback, the training power optimization provides a more significant performance gain than feedback optimization; while conversely for digital feedback, performance is more sensitive to the feedback bit allocation than the training power optimization.  We show that even with low-rate feedback and standard training, ICIC can transform an interference-limited cellular network into a noise-limited one.
\end{abstract}


\section{Introduction}
\emph{Multicell processing}, also called \emph{base station coordination} or \emph{Coordinated Multi-Point (CoMP) transmission/reception}, is an efficient technique to mitigate intercell interference in multicell networks, and it has been attracting significant attention not only from academia \cite{GesHan10JSAC} but also from industry \cite{IrmDro11}. By coordinating the transmission and reception of multiple base stations (BSs), this technique can in principle eliminate intercell interference and transform cellular networks from the familiar interference-limited state to a noise-limited one.

While theoretically appealing, thus far the demonstrated throughput gain of multicell processing in realistic simulations and field trials has been largely disappointing, with typical gains in the 5-15\% range for an LTE-type system \cite{AnnBar10,IrmDro11}.  The large losses versus theory are primarily attributed to the cost of overhead in the form of channel state information (CSI) acquisition and feedback and the difficulty in maintaining the benefits of CoMP in the face of mobility. The need to support at least moderate mobility makes it impossible to amortize these overhead costs over sufficient time periods, as well as causing a fundamental limitation due to the delays inherent in exchanging the various overhead information, which can easily approach a large fraction of the channel coherence time. This motivates directly considering the necessary overhead sources when analyzing multicell processing, optimizing and designing the system for it from the beginning, and thus hopefully achieving higher real-world gains when those overhead losses inevitably occur.

Specifically, we consider a particular multicell processing technique for downlink coordination we term \emph{intercell interference cancellation (ICIC)}. Opposed to full base station cooperation, ICIC is a type of \emph{coordinated single-cell transmission}, as the signal for each user is transmitted from a single BS while the neighboring interfering multi-antenna BSs cancel interference for this user. This means it has low demands on the backhaul capacity, which significantly reduces latency as no inter-BS data sharing is needed and also makes the system robust to synchronization error among BSs. In our previous work \cite{ZhaAnd10JSAC}, we have shown that ICIC can efficiently combat intercell interference and provide performance close to that of multicell processing with full inter-BS data sharing. The goal now is to investigate the performance of ICIC \emph{with a practical CSI model}, where the transmit CSI is obtained through downlink training and uplink feedback. We provide a unified framework to analyze the ICIC system with different CSI assumptions, upon which the training and feedback phases can be optimized.

\subsection{Related Work}

Initial theoretical studies on multi-cell processing assumed full data sharing between BSs and global CSI, which enables global coordination and leads to upper bounds \cite{ShamaiVTC01,Zhang04,Karakayali06a}. To reduce the requirement of backhaul capacity and CSI, clustered coordination schemes have been proposed, including intercell scheduling \cite{ChoAnd08Twc} or more sophisticated interference cancellation approaches \cite{ZhaChe09Twc}. The impact of limited backhaul capacity was investigated in \cite{MarFet07EW,SanSom09IT}. In \cite{BjoZak09Glob, ZakGes10Twc}, distributed BS coordination strategies were proposed based on the virtual signal-to-interference-and-noise (SINR) framework.

The presumed need for a large amount of CSI is a major obstacle for multicell processing. In \cite{RamCai09PIMRC,RamCai09Asilomar}, it was shown that when the CSI overhead is actually taken into account, conventional cellular architectures with no BS coordination are still quite attractive relative to the supposedly superior multicell processing approaches just mentioned. The impact of CSI estimation error in multicell processing was analyzed in \cite{HuhTul10CISS}, while \cite{BhaRao10ICASSP,BhaHea11Tsp} proposed limited feedback techniques for coordinating BSs. These results demonstrated the importance of considering CSI overhead/accuracy, but an accurate characterization of multicell processing systems with both CSI training and feedback, and the corresponding performance optimization, are not yet available.

A related body of research exists for single-cell multiuser MIMO (multiple-input multiple-output) systems, since multicell processing can be regarded as a distributed multiuser MIMO system. Downlink multiuser MIMO requires transmit CSI, and extensive studies have been done on CSI acquisition and optimization for multiuser MIMO. Limited feedback, which provides quantized CSI to the BS and has been successfully implemented in single-user MIMO, has recently received lots of attention in the context of multiuser MIMO \cite{LovHea08JSAC,Jin06IT}. Limited feedback combined with opportunistic user scheduling was investigated in \cite{ShaHas05IT,YooJin07JSAC}. In \cite{ZhaRob09EURASIP,ZhaKou11Tcomm}, the multiuser MIMO system with both delayed and quantized CSI was studied, which shows that single-user MIMO may in fact be preferred over multiuser MIMO with imperfect CSI. In \cite{CaiJin10IT} a comprehensive study of the MIMO broadcast channel with zero-forcing (ZF) precoding considered downlink training and explicit channel feedback, with such overhead optimized in \cite{KobJin09}. In \cite{SanHon10IT}, a similar optimization problem was studied for a beamforming system with limited feedback. Though ICIC bears similarities with multiuser MIMO, the multicell aspect makes CSI training and feedback more decentralized and hence challenging.

One important property of ICIC is that the CSI possessed by the home and neighboring interfering BSs are related to the information signal and interfering signals, respectively. This yields a few important consequences and opportunities. First, the channels from the home and neighboring BSs are statistically independent, which makes the performance easier to analyze. Second, these channels have different path losses, which means the user location will strongly affect the SIR and the quality of channel estimation. Third, the different roles of CSI at the home (related to the information signal) and neighboring BSs (related to interfering signals) motivate adaptively allocating the uplink resource to feed back different channels. Previous works on training and feedback optimization for beamforming and multiuser MIMO in \cite{SanHon10IT,KobJin09} focused on the overhead optimization, which cannot be easily implemented in current systems without significantly altering the standards. In this paper, we consider pilot vs. data power allocation in the downlink training and feedback resource allocation (i.e., feedback symbol power allocation for analogy feedback and feedback bit allocation for digital feedback) in the uplink, which is more amenable to implementation.

\subsection{Contributions}
In this paper, we investigate the performance of ICIC in a 2-cell network, explicitly considering the overheads of CSI training and feedback. The main contributions are as follows.

\textbf{Throughput analysis.} We provide a unified framework to evaluate the average achievable throughput with single-cell beamforming or ICIC at each BS. It shows that the user location has a significant effect, since it affects SIR and the quality of channel estimation. As canceling interference for neighboring users will reduce the signal power for the home user and in addition imperfect CSI degrades the performance, ICIC is not always preferred especially at low SNR. Adaptive strategies are required to improve the system throughput: at low SNR, adaptively switching between single-cell beamforming and ICIC is needed, and ICIC is preferred only when the CSI estimation is accurate and SIR is low; at high SNR, training and feedback can be optimized based on the user locations, fading block length, and the average receive SNR.

\textbf{Training and feedback optimization.} For training optimization, we consider pilot vs. data power allocation, with a tradeoff between the channel estimation accuracy and the downlink throughput. For analog feedback, we optimize the power allocated to the feedback for different BS channels, while for digital feedback, we optimize the number of feedback bits allocated to different BS channels, both of which achieve a balance between the CSI accuracy of the information and interference signals.  The training optimization, since it is broadcast, is common to all users while the feedback optimization is performed individually by each user.
\begin{itemize}
\item \textbf{Training Optimization:} We first provide a sufficient condition under which there is no benefit to optimize the training overhead and simply optimizing the pilot vs. data power tradeoff is adequate. The optimal training power tradeoff is then found, which shows that with block length $T$, and $N_B$ BSs each with $N_t$ antennas, the power allocated to each pilot symbol is proportional to $\sqrt{T/N_t}/N_B$.
\item \textbf{Feedback Optimization:} For periodic analog feedback, it is shown that downlink power tradeoff optimization provides a more significant performance gain than uplink feedback power optimization.  Alternatively, for finite-rate digital feedback, the uplink feedback bit allocation is more important than training optimization.
\end{itemize}
The performance gain of training and feedback optimization is demonstrated by simulation, which shows that ICIC with CSI training and feedback provides significant average and edge throughput gains over conventional single-cell beamforming. Specifically, with proportional uplink/downlink transmit power (for analog feedback) or with feedback bits proportional to the uplink channel capacity (for digital feedback), ICIC with standard training and feedback provides performance approaching that with perfect CSI, and successfully transforms an interference-limited cellular network into a noise-limited one, even accounting for the necessary overhead.


\section{System Model}\label{Sec:Model}

\subsection{Signal Model}
We consider a multicell network with $N_B$ BSs each with $N_t$ antennas, and there is one active single-antenna user in each cell. Universal frequency reuse is assumed. The objective is to design an efficient multicell processing strategy to suppress intercell interference. To retain tractability, we will focus on the 2-cell network as shown in Fig. \ref{fig:2cell}. The BS and user in the $i$-th cell are indexed by $i$, while the BS and user in the other cell are indexed by $\bar{i}=\mbox{mod}(i,2)+1$ for $i=1,2$.

We focus on the downlink transmission. For the data symbol transmission, the discrete baseband signal received at the $i$-th user ($i=1,2$) is given as
\begin{equation}
y_i=\sqrt{P_dL_{i,i}}\mathbf{h}_{i,i}^*\mathbf{f}_ix_i+\sqrt{P_dL_{i,\bar{i}}}\mathbf{h}_{i,\bar{i}}^*\mathbf{f}_{\bar{i}}x_{\bar{i}}+z_i,
\end{equation}
where $\mathbf{a}^*$ is the conjugate transpose of a vector $\mathbf{a}$ and
\begin{itemize}
\item $x_i$ is the transmit signal from the $i$-th BS for the $i$-th user, with the power constraint $\mathbb{E}[|x_i|^2]=1$.
\item $z_i$ is the complex white Gaussian noise with zero mean and unit variance, i.e., $z_i\sim\mathcal{CN}(0,1)$.
\item $P_d$ is the transmit power for data symbols and $L_{i,j}$ is the pathloss given by $L_{i,j}=\eta\left(D_0/d_{i,j}\right)^\alpha$, where $D_0$ is the reference distance, $\eta$ is a unitless constant that depends on the antenna characteristics, and $d_{i,j}$ is the distance between user $i$ and BS $j$. In the following, we set $D_0=R$, so $\eta P_d$ is the average received SNR at the cell edge.
\item $\mathbf{h}_{i,j}$ is the $N_t\times{1}$ channel vector from the BS $j$ to user $i$, where each component is i.i.d. $\mathcal{CN}(0,1)$. We consider a block fading model, where the channel is constant over each block of length $T$ and is independent for different blocks. 
\end{itemize}

The vector $\mathbf{f}_{i}$ is the precoding vector at BS $i$, $i=1,2$. It is normalized and is designed based on the available CSI at the BS. Although the main focus is on ICIC (with the precoder denoted as $\bff_{i,\mathrm{IC}}$), we also consider conventional single-cell beamforming (with the precoder denoted as $\bff_{i,\mathrm{BF}}$):
\begin{itemize}
\item \textbf{Single-cell Eigen-beamforming:} The precoding vector is the channel direction, i.e., for the $i$-th user $\mathbf{f}_{i,\mathrm{BF}}=\mathbf{h}_{i,i}/\|\mathbf{h}_{i,i}\|$. Therefore, the signal term is distributed as $|\mathbf{f}_{i,\mathrm{BF}}^*\mathbf{h}_{i,i}|^2\sim\chi^2_{2N_t}$, where $\chi^2_n$ denotes the chi-square random variable with $n$ degrees of freedom. When applying eigen-beamforming in multicell networks, each user will suffer interference from other cells.

\item \textbf{Intercell Interference Cancellation (ICIC):} Taking cell 1 as an example, to cancel its interference for users in cell 2, 3, $\cdots$, $N_B$, ($N_B\leq{N_t}$), and also to maximize the desired signal power $|\mathbf{f}_{1,\mathrm{IC}}^*\mathbf{h}_{1,1}|^2$, the precoding vector $\mathbf{f}_{1,\mathrm{IC}}$ is chosen in the direction of the projection of vector $\mathbf{h}_{1,1}$ on the nullspace of vectors $\hat{\mathbf{H}}=[\mathbf{h}_{2,1},\mathbf{h}_{3,1},\cdots,\mathbf{h}_{N_B,1}]$ \cite{JinAnd10Tcomm}, i.e., the precoding vector is the normalized version of the vector $\mathbf{w}_1^{(1)}=\left(\mathbf{I}-\hat{\mathbf{H}}\left(\hat{\mathbf{H}}^*\hat{\mathbf{H}}\right)^{-1}\hat{\mathbf{H}}^*\right)\mathbf{h}_{1,1}$.
    From \cite{JinAnd10Tcomm}, we have the distribution of the signal power as $|\mathbf{f}_{1,\mathrm{IC}}^*\mathbf{h}_{1,1}|^2\sim\chi^2_{2(N_t-(K-1))}$. Note that the precoder design only requires local CSI, i.e., BS $k$ only needs $\bh_{i,j}$ with $j=k$, and $i=1,2$.
\end{itemize}

Assuming user $j$ selects the transmission strategy $s_j\in\{\mathrm{IC},\mathrm{BF}\}$, $j=1,2$, the receive signal-to-interference-plus-noise ratio (SINR) for user $i$ is
\begin{equation}\label{eq:SINR}
\mbox{SINR}_i(s_1,s_2)=\frac{P_dL_{i,i}|\mathbf{h}_{i,i}^*\mathbf{f}_{i,s_i}|^2}{1+P_dL_{i,\bar{i}}|\mathbf{h}_{i,\bar{i}}^*\mathbf{f}_{\bar{i},s_{\bar{i}}}|^2}.
\end{equation}
The SINR not only depends on the transmission strategy pair, but also depends on the available CSI, which affects the precoder design.

\subsection{The CSI Model}
To enable transmit precoding, downlink CSI is required at the BS. In TDD (Time Division Duplex) systems with channel reciprocity, transmit CSI can be obtained through uplink training. For systems without channel reciprocity, such as FDD (Frequency Division Duplex) systems, downlink training and uplink feedback are normally applied to provide transmit CSI, which is our focus in the paper. Each fading block of length $T$ is divided into three phases: a downlink training phase of $T_t$ channel uses, an uplink feedback phase of $T_{fb}$ channel uses, and the data transmission phase of $T_d$ channel uses. CSI training is discussed in this subsection, while feedback will be treated in Section \ref{Sec:TrainingFB}.

The CSI at each user is obtained through downlink training with dedicated pilot symbols. We consider orthogonal training, where the training phase spans $T_t$ ($T_t\geq N_BN_t$) channel uses, using orthogonal training sequences $\{\phi_0,\ldots,\phi_{N_BN_t-1}\}$, with $\phi_i\in\mathbb{C}^{T_t\times1}$. The set of training sequences is partitioned into $N_B$ disjoint groups each with $N_t$ sequences, denoted as $\Phi_i$ for the $i$th BS, $i=1,2,\ldots,N_B$. The power scaling factor is $\sqrt{\frac{T_t}{N_t}P_t}$ so the transmit power for the pilot symbols from each BS is $T_tP_t$, which sets the power constraint for each pilot symbol to be $P_t$. For simplicity, we normalize $T$ and $T_t$ as $\overline{T}\triangleq\frac{T}{N_t}$, and $\overline{T}_t\triangleq\frac{T_t}{N_t}$.

Different from conventional single-cell processing systems, we assume that each user estimates CSI from both its home BS and the neighboring BS. The user $i$ estimates the channel from BS $j$ based on the observation
\begin{equation}\label{eq:s_k}
\bs_{i,j}=\sqrt{\overline{T}_tP_tL_{i,j}}\bh_{i,j}+\bz_i,\quad i,j=1,2,
\end{equation}
corresponding to the common training channel output, where $\bz_i\sim\mathcal{CN}(0,\bI_{N_t})$. If $i=j$, $\bs_{i,j}$ corresponds to the pilot from the home BS; otherwise, it corresponds to the pilot from the neighboring BS. While the channel estimation quality of the home BS determines the information signal power, the channel estimation of neighboring cells determines the residual intercell interference. Therefore, the training design for multicell processing is quite different from conventional single-cell processing systems.

The minimum mean square error (MMSE) estimate of $\bh_{i,j}$ given the observation $\bs_{i,j}$ is \cite{Kay93I}
\begin{equation}\label{eq:CHest}
\tilde{\bh}_{i,j}=\mathbb{E}[\bh_{i,j}\bs_{i,j}^H]\mathbb{E}[\bs_{i,j}\bs_{i,j}^H]^{-1}\bs_{i,j}
=\frac{\sqrt{\overline{T}_tP_tL_{i,j}}}{1+\overline{T}_tP_tL_{i,j}}\bs_{i,j},\quad i,j=1,2.
\end{equation}
The channel $\bh_{i,j}$ can be written in terms of the estimate $\tilde{\bh}_{i,j}$ and the estimation noise $\bn_{i,j}$ as
\begin{equation}\label{eq:H-Hest}
\bh_{i,j}=\tilde{\bh}_{i,j}+\bn_{i,j},\quad i,j=1,2.
\end{equation}
With the MMSE estimator \cite{Kay93I}, $\bn_{i,j}$ is independent of the estimate and is zero-mean Gaussian with covariance $\sigma_{i,j}^2\bI_{N_t}$ with $\sigma_{i,j}^2=\frac{1}{1+\overline{T}_tP_tL_{i,j}}$, while $\tilde{\bh}_{i,j}$ is with covariance $\kappa^2_{i,j}\bI_{N_t}$, $\kappa_{i,i}^2=1-\sigma_{i,j}^2$.

Considering the two downlink transmit phases, i.e., the training phase and the data transmission phase, we have the following constraints:
\begin{align}
\overline{T}_t+\overline{T}_d=\overline{T}-\overline{T}_{fb},\quad
P_t\overline{T}_t+P_d\overline{T}_d=P^{dl}(\overline{T}-\overline{T}_{fb}),
\end{align}
where $\overline{T}_d\triangleq\frac{T_d}{N_t}$, $\overline{T}_{fb}\triangleq\frac{T_{fb}}{N_t}$, and $P^{dl}$ is the average power constraint in the downlink. The training optimization with these constraints will be investigated in Section \ref{Sec:Training}.

With the receive SINR given in \eqref{eq:SINR}, treating intercell interference as additive white Gaussian noise, we are interested in the following average achievable throughput
\begin{equation}\label{eq:Rate}
R_i(s_1,s_2)=\mathbb{E}\left[\log_2\left(1+\mbox{SINR}_i\right)\right],\quad i=1,2.
\end{equation}

\begin{remark}
To achieve this throughput, an additional dedicated training round is needed after the uplink feedback phase for each user to get its received SINR value \cite{CaiJin10IT}. For convenience, we ignore this dedicated training round and assume there is a genie who provides each user the actual SINR, and this \emph{genie-aided} throughput will be used as the performance metric throughout the paper. As the capacity of this kind of interference channel is unknown even with perfect CSI, our focus is on the achievable throughput with specific training and feedback methods.
\end{remark}

Considering the overhead due to training and feedback, the effective achievable throughput is
\begin{equation}
\overline{R}_i(s_1,s_2)=\left(1-\frac{\overline{T}_t+\overline{T}_{fb}}{\overline{T}}\right)R_i(s_1,s_2), \quad i= 1,2.
\end{equation}
We will use $R_i$, $R_{i,\mathrm{T}}$, $R_{i,\mathrm{aFB}}$, and $R_{i,\mathrm{dFB}}$ to denote the throughput of user $i$ with perfect CSI, downlink training, training and analog feedback, and training and digital feedback, respectively.

\subsection{Auxiliary Results}
This subsection provides useful results that are used for throughput analysis throughout the paper. The average achievable throughput in \eqref{eq:Rate} can be in general expressed as $\mathbb{E}[\log_2(1+X)]$ with $X\triangleq\frac{Z}{1+Y}$, where $Z$ and $Y$ are the signal power and the interference power, respectively. With independent Rayleigh fading channels, the following results from \cite{ZhaAnd10JSAC} can be used for throughput analysis for different systems considered in the paper.

\begin{lemma}\label{lemma:Rate}
Consider the random variable $X\triangleq\frac{\alpha Z}{1+Y}$, where $\alpha>0$, and $Z$ and $Y$ are independent.
\begin{enumerate}
\item If $Z\sim\chi_{2M}^2$ and $Y=0$, then
    \begin{align}\label{eq:Rate_BF}
    \mathcal{R}^{(1)}(\alpha,M)\triangleq\mathbb{E}_X\left[\log_2\left(1+X\right)\right]
    =\log_2(e)e^{1/\alpha}\sum_{k=0}^{M-1}\frac{\Gamma(-k,1/\alpha)}{\alpha^k}.
    \end{align}
\item If $Z\sim\chi^2_{2M}$, $Y\sim\beta\cdot\chi^2_2$, then
\begin{align}\label{eq:Rate_I2}
\mathcal{R}^{(2)}(\alpha,\beta,M)\triangleq\mathbb{E}_X\left[\log_2(1+X)\right]=\log_2(e)\sum_{i=0}^{M-1}\sum_{l=0}^i\frac{\alpha^{l+1-i}}{\beta(i-l)!}\cdot{I_1}\left(\frac{1}{\alpha},\frac{\alpha}{\beta},i,l+1\right),
\end{align}
where $I_1(a,b,m,n)=\int_0^\infty\frac{x^me^{-ax}}{(x+b)^n(x+1)}\mbox{d}x$, with a closed-form expression given in \cite{ZhaAnd10JSAC}.
\item If $Z\sim\chi^2_{2M}$, $Y=\beta_1{Y_1}+\beta_2{Y_2}$ with $Y_1\sim\chi^2_{2}$, $Y_2\sim\chi^2_{2}$, and they are mutually independent,
    \begin{align}\label{eq:Rate_3cell}
&\mathcal{R}^{(3)}(\alpha,\beta_1,\beta_2,M)\triangleq\mathbb{E}_X\left[\log_2(1+X)\right]\notag\\
=&\log_2(e)\sum_{i=0}^{M-1}\sum_{l=0}^i\frac{\alpha^{l-i+1}}{(\beta_1-\beta_2)(i-l)!}\left[I_1\left(\frac{1}{\alpha},\frac{\alpha}{\beta_1},i,l+1\right)-I_1\left(\frac{1}{\alpha},\frac{\alpha}{\beta_2},i,l+1\right)\right].
\end{align}
\end{enumerate}
\end{lemma}

\section{ICIC with CSI Training ($T_t>0, T_{fb}=0$)}\label{Sec:Training}
In this section, we consider the 2-cell network with downlink CSI training and assume that the BS has direct access to the CSI estimated at the user. The results developed following this assumption can be applied to the TDD system where the transmit CSI can be obtained through uplink training, and they will also be applied to the system with both training and feedback in Section \ref{Sec:TrainingFB}.

\subsection{Throughput Analysis and Adaptive Transmission}
One unique property of wireless networks is the spatial distribution of different nodes, which causes multiple order of magnitude fluctuation in the information signal power and the interference level. Three typical scenarios of user locations in the 2-cell network are shown in Fig. \ref{fig:2cell}. As shown in \cite{ZhaAnd10JSAC}, concerning the sum throughput, it is not always optimal to apply ICIC, especially for low to medium receive SNRs (e.g., for both users in scenario (b), or user 2 in scenario (c) in Fig. \ref{fig:2cell}), as cancelling interference for the neighboring user will reduce the signal power for the home user. Therefore, it is necessary to switch between single-cell beamforming and ICIC based on user locations. When considering CSI training, as the pilot symbols received from the home BS and the neighboring BS come from different propagation paths, the user location now will also affect the accuracy of the estimation for different channels.

For adaptive transmission strategy selection, we first derive the average achievable throughput, given in \emph{Proposition \ref{Prop:Training}}.

\begin{proposition}\label{Prop:Training}
The average achievable throughput of user $i$ ($i=1,2$) with estimated CSI is approximated by
\begin{equation}\label{eq:Rate_CHest}
R_{i,\mathrm{T}}(s_1,s_2)\approx\left\{\begin{array}{ll}
\mathcal{R}^{(2)}(\kappa_{i,i}^2P_dL_{i,i},P_dL_{i,\bar{i}},N_t) & (s_i,s_{\bar{i}})=(\mathrm{BF},\mathrm{BF})\\
\mathcal{R}^{(2)}(\kappa_{i,i}^2P_dL_{i,i},\sigma_{i,\bar{i}}^2P_dL_{i,\bar{i}},N_t) & (s_i,s_{\bar{i}})=(\mathrm{BF},\mathrm{IC})\\
\mathcal{R}^{(2)}(\kappa_{i,i}^2P_dL_{i,i},\sigma_{i,\bar{i}}^2P_dL_{i,\bar{i}},N_t-1) & (s_i,s_{\bar{i}})=(\mathrm{IC},\mathrm{IC})\\
\mathcal{R}^{(2)}(\kappa_{i,i}^2P_dL_{i,i},P_dL_{i,\bar{i}},N_t-1) & (s_i,s_{\bar{i}})=(\mathrm{IC},\mathrm{BF})
\end{array}\right.
\end{equation}
where $\mathcal{R}^{(2)}$ is given in \eqref{eq:Rate_I2}.
\end{proposition}
\begin{proof}
See Appendix \ref{App:PropTraining}.
\end{proof}

To maximize the sum throughput, the preferred transmission strategy pair is then determined as
\begin{equation}\label{eq:Adaptive}
(s^\star_1,s^\star_2)=\arg\max_{s_1,s_2\in\{\mathrm{BF},\mathrm{IC}\}}R_{1,\mathrm{T}}(s_1,s_2)+R_{2,\mathrm{T}}(s_1,s_2).
\end{equation}

Fig. \ref{fig:User1Edge} shows the simulation and calculation results with perfect CSI and CSI training, where user 1 is at the location $(-0.1R,0)$ in Fig. \ref{fig:2cell}. The results for perfect CSI come from \cite{ZhaAnd10JSAC}. It shows that the approximation \eqref{eq:Rate_CHest} is accurate and the estimated CSI degrades the performance. Specifically, with CSI training, when user 2 is in the interior (near $(R,0)$), there is little or no performance gain for user 1 to perform ICIC for user 2, which is quite different from the perfect CSI case. This is because when user 2 is close to BS 2, the received pilot signal power from BS 1 becomes weak, which reduces the estimation accuracy and results a high level of residual interference from BS 1.

In Fig. \ref{fig:ModePlot}, we plot the preferred transmission strategy pairs for different user locations, with edge SNR 4 dB. The preferred strategy pair is determined by \eqref{eq:Adaptive}. Comparing the system with perfect CSI and with CSI training, the operating region of each strategy pair changes significantly. With CSI training, the strategy $\mathrm{IC}$ is picked at a given BS only when the neighboring user is very close to the cell edge, where the CSI estimation is accurate and the performance gain from ICIC is high.


\subsection{High SNR Performance and Training Optimization}\label{Sec:TrainingOpt}
In this subsection, we investigate the high-SNR performance of ICIC. While adaptive transmission is required at low SNR, ICIC is always preferred at high SNRs due to its capability to suppress intercell interference, as the multicell network is interference-limited at high SNR.

With perfect CSI, intercell interference is completely cancelled, and the average achievable throughput for the $i$-th user ($i=1,2$) at high SNR can be approximated as $R_{i}\approx\mathbb{E}\left[\log_2\left(L_{i,i}P_d\chi_{2(N_t-1)}^2\right)\right]$, where $\chi_{2n}^2$ denotes a chi-square random variable with $2n$ degrees of freedom. As
$\mathbb{E}\left[\log\chi_{2n}^2\right]=\psi(n)$, where $\psi(\cdot)$ is Euler's digamma function that satisfies $\psi(m)=\psi(1)+\sum_{l=1}^{m-1}\frac{1}{l}$ for positive integers $m$ and $\psi(1)\approx-0.577215$, we have
\begin{equation}\label{eq:R_CSI}
R_i\approx\log_2P_dL_{i,i}+\log_2e\cdot\psi(N_t-1)=\log_2\left(P_dL_{i,i}e^{\psi(N_t-1)}\right),
\end{equation}
which depends on the user location and $P_d$.

When CSI is obtained through downlink training, the channel estimation error will degrade the system performance. The following lemma shows the impact of CSI training at high SNR.
\begin{lemma}[Throughput loss due to training]\label{lemma:rateloss}
At high SNR (assuming $P_t,P_d\rightarrow\infty$ with $\frac{P_d}{P_t}=\nu$), the throughput loss of ICIC due to CSI training for user $i$ ($i=1,2$) is
\begin{equation}
\Delta{R}_{i,\mathrm{T}}=R_{i}-R_{i,\mathrm{T}}\approx\mathbb{E}\left[\log\left(1+\frac{\nu}{\overline{T}_t}\chi_2^2\right)\right]=\mathcal{R}^{(1)}\left(\frac{\nu}{\overline{T}_t},1\right),
\end{equation}
where $\mathcal{R}^{(1)}$ is given in \eqref{eq:Rate_BF}.
\end{lemma}
\begin{proof}
See Appendix \ref{App:rateloss}.
\end{proof}

The throughput loss due to training is a constant at high SNR and depends on the ratio of $\nu$ and $\overline{T}_t$. It decreases as $\nu$ decreases (more power is allocated to training symbols) and/or the training period $\overline{T}_t$ increases. For example, if $P_d=P_t$ and $\overline{T}_t=N_B$ with $N_B=2$, then $\Delta{R}_{i,\mathrm{T}}\approx R^{(1)}(1/2,1)\approx0.52\mbox{ bps/Hz}$. This rate loss is negligible at high SNR.

Based on the throughput loss and applying Jensen's inequality $\mathcal{R}^{(1)}\left(\frac{\nu}{\overline{T}_t},1\right)\leq\log_2\left(1+\frac{\nu}{\overline{T}_t}\right)$, we can obtain a high-SNR approximation for the achievable throughput with CSI training for user $i$
\begin{align}
R_{i,\mathrm{T}}\approx R_i-\log_2\left(1+\frac{P_d}{P_t}\frac{1}{\overline{T}_t}\right)
=\log_2\frac{L_{i,i}e^{\psi(N_t-1)}}{P_d^{-1}+\overline{T}^{-1}_tP_t^{-1}}.\label{eq:RApprox}
\end{align}
From this expression, the achievable throughput with CSI training depends on the power allocation between $P_t$ and $P_d$, and the training overhead $\overline{T}_t$, which motivates to optimize downlink training to improve the performance.

In general, the training optimization involves optimizing both $P_t$ and $T_t$. We first consider optimizing $P_t$. We assume different BSs will have the same $P_d$ and $P_t$, which is desirable for practical system implementation. Considering user $i$ ($i=1,2$), the training power allocation problem is $(P_t^\star,P_d^\star)=\arg\max_{\overline{T}_tP_t+\overline{T}_dP_d=\overline{T}P^{dl}}R_{i,\mathrm{T}}$.
With the approximation \eqref{eq:RApprox}, the problem is reformulated as
\begin{equation}\label{eq:PtOptProblem}
(P_t^\star,P_d^\star)=\arg\min_{\overline{T}_tP_t+\overline{T}_dP_d=\overline{T}P}\frac{1}{P_d}+\frac{1}{\overline{T}_tP_t}.
\end{equation}
From this formulation, the training power allocation is the same for different BSs and is independent of user locations. This fits the practical requirement, as each BS transmits the common pilot symbols for all users in the cell.

The minimization problem \eqref{eq:PtOptProblem} is a convex optimization problem, and following the KKT (Karush-Kuhn-Tucker) condition we can get the solution as
\begin{align}
P_d^\star=\frac{\overline{T}P^{dl}}{\sqrt{\overline{T}_d}(\sqrt{\overline{T}_d}+1)},\quad
P_t^\star=\frac{\overline{T}P^{dl}}{{\overline{T}_t}(\sqrt{\overline{T}_d}+1)}\label{eq:PtOpt}.
\end{align}

Next, we consider optimizing both $P_t$ and $\overline{T}_t$ ($N_B\leq\overline{T}_t<\overline{T}$). For a given $\overline{T}_t$, $P_t^\star$ is obtained as in \eqref{eq:PtOpt}. With the high-SNR approximation \eqref{eq:RApprox}, and substituting $P_t^\star$ and $P_d^\star$, we have the following approximation for the equivalent throughput considering training overhead
\begin{align}
\overline{R}_{i,T}=\left(1-\frac{\overline{T}_t}{\overline{T}}\right)R_{i,T}\approx
\left(1-\frac{\overline{T}_t}{\overline{T}}\right)\log_2\frac{L_{i,i}e^{\psi(N_t-1)}\overline{T}P^{dl}}{(\sqrt{\overline{T}-\overline{T}_t}+1)^2},\quad i=1,2,
\end{align}
which is concave in $T_t$, so the optimal length $T_t^\star$ ($N_B\leq\overline{T}_t<\overline{T}$)
can be found by a line search, e.g., using the bisection method. However, the solution will depend on the location of user $i$ through $L_{i,i}$. This means that we cannot simultaneously optimize $\overline{T}_t$ for users at different locations, which conflicts the common pilot design rule. Fortunately, this will not cause an issue, as the optimal $\overline{T}_t$ is obtained at its minimum possible value $N_B$ under a mild condition, stated in the following lemma.

\begin{lemma}\label{lemma:SuffCond}
A sufficient condition for $\overline{T}_t^\star=N_B$ is
\begin{equation}\label{eq:SuffCond}
P^{dl}L_{ii}>\exp\left[\frac{\sqrt{\overline{T}-N_B}}{\sqrt{\overline{T}-N_B}+1}
+2\log\left(\sqrt{1-\frac{N_B}{\overline{T}}}+\frac{1}{\overline{T}}\right)-\psi(N_t-1)\right]\quad i=1,2.
\end{equation}
\end{lemma}
\begin{proof}
See Appendix \ref{App:SuffCond}.
\end{proof}

\begin{remark}
This result shows that pilot/data power allocation is sufficient for training optimization under condition \eqref{eq:SuffCond}, and there is no need to optimize $\overline{T}_t$.
\begin{enumerate}
\item This is a condition on the average receive SNR at each user. If the average edge SNR is greater than the threshold, then all the users in the cell satisfy this condition.
\item The threshold at the right-hand side of \eqref{eq:SuffCond} depends on $N_B$, $N_t$, and $T$, and it is a decreasing function of $N_B$ and $N_t$. This condition is easy to satisfy. For example, if $T=100$, $N_B=2$, we need $P^{dl}L_{ii}>0.87$ dB for $N_t=4$, and $P^{dl}L_{ii}>-3.24$ dB for $N_t=8$.
\item A similar result is also shown in \cite{KobCai08ISIT} for single-cell multiuser MIMO systems. In fact, for a single-user MIMO system, there is never a need to optimize $T_t$ once $P_t$ is optimized \cite{HasHoc03IT}.
\end{enumerate}
\end{remark}

As condition \eqref{eq:SuffCond} is easy to satisfy, we will only consider pilot power optimization in the following, and assume $\overline{T}_t=N_B$. When $\overline{T}$ is very large, for fixed $\overline{T}_t$, we have $\overline{T}_d\rightarrow \overline{T}$, and from \eqref{eq:PtOpt} we have $P_d^\star\approx P^{dl}$, $P_t^\star\approx \frac{\sqrt{\overline{T}}}{N_B}P^{dl}$.
So the power allocated to training symbols is proportional to $\frac{\sqrt{\overline{T}}}{N_B}$, i.e., for long blocks, more power will be allocated to each training symbol.

\section{ICIC with CSI Training and Feedback ($T_t>0,T_{fb}>0$)}\label{Sec:TrainingFB}
In this section, we consider the CSI model with both training and feedback. Two types of feedback are discussed: analog feedback and digital feedback. For each feedback type, we first derive the average achievable throughput, and then optimize the training and feedback phases.

\subsection{Training with Analog Feedback}
We first consider \emph{analog feedback} where the estimated CSI at each user is fed back to the BS using unquantized and uncoded QAM \cite{MarHoc06Tsp,SamMan06Tcomm}. The uplink feedback channel is assumed to be an unfaded AWGN channel as in \cite{KobCai08ISIT,KobJin09}. As each user needs to feed back CSI for $N_B$ BSs, we divide the feedback block $T_{fb}$ into $N_B$ equal-length sub-blocks. During the $j$th sub-block, at user $i$, $i=1,2$, the estimated CSI, or equivalently a scaled version of the received training vector $\bs_{i,j}$ ($j=1,2$), given in \eqref{eq:s_k}, is modulated by a $\frac{T_{fb}}{N_B}\times N_t$ unitary spreading matrix with the power constraint $\frac{T_{fb}}{N_B}P_{fb,ij}$ \cite{MarHoc06Tsp}. We assume orthogonal feedback, so $\frac{T_{fb}}{N_B}\geq N_BN_t$, i.e., $T_{fb}\geq N_B^2N_t$. Although the feedback can be received by both BSs, the home BS $i$ is responsible for the final channel estimation as it is closer to user $i$, i.e., BS $i$ will estimate both $\bh_{i,1}$ and $\bh_{i,2}$ and will pass the estimation to the neighboring BS over the backhaul link. In the following discussion, we fix ${T}_{fb}$ and focus on the power allocation $(P_{fb,i1},P_{fb,i2})$, with the constraint
\begin{align}\label{eq:aFB_con}
\frac{T_{fb}}{N_B}P_{fb,i1}+\frac{T_{fb}}{N_B}P_{fb,i2}=T_{fb}P^{ul}\quad \mbox{ or }\quad
P_{fb,i1}+P_{fb,i2}=P^{ul}N_B,
\end{align}
where $P^{ul}$ is the uplink transmit power constraint.

As the uplink channel is modeled as an unfaded AWGN channel with pathloss, the received feedback vector at the $i$th BS after de-spreading is
\begin{align}
\bg_{i,j}&=\frac{\sqrt{\frac{{T}_{fb}}{N_B}P_{fb,ij}L_{i,i}}}{\sqrt{1+\overline{T}_tP_tL_{i,j}}}\bs_{i,j}+\bw_{i,j}
=\frac{\sqrt{\frac{T_{fb}}{N_B}P_{fb,ij}L_{i,i}}}{\sqrt{1+\overline{T}_tP_tL_{i,j}}}\left(\sqrt{\overline{T}_tP_tL_{i,j}}\bh_{i,j}+\bz_{i,j}\right)+\bw_{i,j}\notag\\
&=\frac{\sqrt{\frac{T_{fb}}{N_B}P_{fb,ij}L_{i,i}\overline{T}_tP_tL_{i,j}}}{\sqrt{1+\overline{T}_tP_tL_{i,j}}}\bh_{i,j}+\tilde{\bw}_{i,j},\quad j=1,2,
\end{align}
where $\tilde{\bw}_{i,j}$ is the equivalent noise, and $\tilde{\bw}_{i,j}\sim\mathcal{CN}(0,\tilde{\sigma}^2_{i,j})$ with $\tilde{\sigma}^2_{i,j}=\frac{\frac{T_{fb}}P_{fb,ij}L_{i,i}}{1+\overline{T}_tP_tL_{i,j}}+1$. If $i=j$, the feedback is for the home BS channel, which determines the signal power; if $i\neq j$, the feedback is for the neighboring BS channel, which is related to the interference level. This motivates the feedback power allocation with the constraint \eqref{eq:aFB_con}.

The MMSE estimate of the channel vector is
\begin{equation}
\hat{\bh}_{i,j}=\frac{\sqrt{\frac{T_{fb}}{N_B}P_{fb,ij}L_{i,i}\overline{T}_tPL_{i,j}}}{\sqrt{1+\overline{T}_tP_tL_{i,j}}\left(\frac{T_{fb}}{N_B}P_{fb}L_{i,i}+1\right)}\bg_{i,j}.
\end{equation}
Then the actual channel vector $\bh_{i,j}$ can be written as $\bh_{i,j}=\hat{\bh}_{i,j}+\hat{\bee}_{i,j}$,
where $\hat{\bh}_{i,j}$ and $\hat{\bee}_{i,j}$ are independent with variances
$\hat{\kappa}^2_{i,j}=\frac{\overline{T}_tP_tL_{i,j}\cdot \frac{T_{fb}}{N_B}P_{fb,ij}L_{i,i}}{(1+\overline{T}_tP_tL_{i,j})\left(1+\frac{T_{fb}}{N_B}P_{fb,ij}L_{i,i}\right)}$ and $\hat{\sigma}^2_{i,j}=1-\hat{\kappa}^2_{i,j}$, respectively.

The precoding vectors are designed assuming that $\hat{\bh}_{i,j}$ ($i,j=1,2$) are the actual CSI. As the distribution of $\hat{\bh}_{i,j}$ is similar to $\tilde{\bh}_{i,j}$ with different variances, following the same derivation for the case with training only, we can get the following proposition.

\begin{proposition}
The average achievable throughput of user $i$ with training and analog feedback is approximated as
\begin{equation}\label{eq:Rate_AFB}
R_{i,\mathrm{aFB}}(s_1,s_2)\approx\left\{\begin{array}{ll}
\mathcal{R}^{(2)}(\hat{\kappa}_{i,i}^2P_dL_{i,i},P_dL_{i,\bar{i}},N_t) & (s_i,s_{\bar{i}})=(\mathrm{BF},\mathrm{BF})\\
\mathcal{R}^{(2)}(\hat{\kappa}_{i,i}^2P_dL_{i,i},\hat{\sigma}_{i,\bar{i}}^2P_dL_{i,\bar{i}},N_t) & (s_i,s_{\bar{i}})=(\mathrm{BF},\mathrm{IC})\\
\mathcal{R}^{(2)}(\hat{\kappa}_{i,i}^2P_dL_{i,i},\hat{\sigma}_{i,\bar{i}}^2P_dL_{i,\bar{i}},N_t-1) & (s_i,s_{\bar{i}})=(\mathrm{IC},\mathrm{IC})\\
\mathcal{R}^{(2)}(\hat{\kappa}_{i,i}^2P_dL_{i,i},P_dL_{i,\bar{i}},N_t-1) & (s_i,s_{\bar{i}})=(\mathrm{IC},\mathrm{BF})
\end{array}\right.
\end{equation}
where $\mathcal{R}^{(2)}$ is given in \eqref{eq:Rate_I2}.
\end{proposition}
This result can be used for adaptive transmission strategy selection.

In the following, we will optimize both training (optimizing $(P_d,P_t)$) and feedback (optimizing $(P_{fb,i1},P_{fb,i2})$ for $i=1,2$). We assume that $T_{fb}$ is fixed, as the modification of $T_{fb}$ will affect the uplink traffic channel while our discussion focuses on the downlink transmission. A systematic design of both downlink and uplink transmissions is beyond the scope of this paper.

\subsubsection{Training Optimization}
We first consider training optimization. At high SNR, similar to \emph{Lemma \ref{lemma:rateloss}}, the rate loss due to training and analog feedback can be approximated as
\begin{align}
R_i-R_{i,\mathrm{aFB}}\approx
\log_2\left(1+L_{i,\bar{i}}P_d\left(\frac{1}{\frac{T_{fb}}{N_B}P_{fb,i\bar{i}}L_{i,i}}+\frac{1}{\overline{T}_tP_tL_{i,\bar{i}}}\right)\right),\label{eq:Rbound_aFB}
\end{align}
which is a constant rate loss if $\frac{P_d}{P_t}$ and $\frac{P_d}{P_{fb,i\bar{i}}}$ are constants.

Substituting \eqref{eq:R_CSI}, we get the following high-SNR approximation for $R_{i,\mathrm{aFB}}$
\begin{equation}\label{eq:41}
R_{i,\mathrm{aFB}}\approx\log_2\frac{L_{i,i}e^{\psi(N_t-1)}}{P_d^{-1}+\left(\bar{T}_tP_t\right)^{-1}+
\left(\frac{T_{fb}}{N_B}\frac{L_{i,i}}{L_{i,\bar{i}}}P_{fb,i\bar{i}}\right)^{-1}}.
\end{equation}
For given $T_{fb}$ and $P_{fb,i\bar{i}}$, the throughput maximization problem is equivalent to
\begin{equation}
\min_{\overline{T}_tP_t+\overline{T}_dP_d=(\overline{T}-\overline{T}_{fb})P^{dl}}\frac{1}{P_d}+\frac{1}{\overline{T}_tP_t},
\end{equation}
which is independent of the feedback allocation. This problem is the same as \eqref{eq:PtOpt}, and we get the following results
\begin{align}
P_d^\star=\frac{(\overline{T}-\overline{T}_{fb})P^{dl}}{\sqrt{\overline{T}_d}(\sqrt{\overline{T}_d}+1)},\quad
P_t^\star=\frac{(\overline{T}-\overline{T}_{fb})P^{dl}}{{\overline{T}_t}(\sqrt{\overline{T}_d}+1)}.\label{eq:P_tOptaFB}
\end{align}
The solution depends only on the intervals of different transmission phases, i.e., $T_t$, $T_d$, and $T_{fb}$.

\subsubsection{Feedback Optimization}
Next, we consider feedback optimization, i.e., optimizing $(P_{fb,i1},P_{fb,i2})$ for $i=1,2$. Note that the uplink feedback optimization is done individually for each user, while the downlink training optimization is the same for all users. The feedback optimization is over the following approximation for the average SINR
\begin{equation}\label{eq:SINRapprox_aFB}
\overline{\mbox{SINR}}_i\approx\frac{P_dL_{i,i}\hat{\kappa}_{i,i}^2(N_t-1)}{1+P_dL_{i,\bar{i}}\hat{\sigma}_{i,\bar{i}}^2}, \quad i=1,2.
\end{equation}
This is reasonable as $\log_2(1+\overline{\mbox{SINR}}_i)$ gives an upper bound on the average achievable rate for user $i$. From \eqref{eq:SINRapprox_aFB}, the feedback power allocation problem at user $i$ can be stated as
\begin{equation}\label{eq:aFB_PA}
(P^\star_{fb,ii},P^\star_{fb,i\bar{i}})=\arg\max_{P_{fb,ii}+P_{fb,i\bar{i}}=N_BP^{ul}}\frac{\hat{\kappa}_{i,i}^2}{1+P_dL_{i,\bar{i}}\hat{\sigma}_{i,\bar{i}}^2}.
\end{equation}
Increasing $P_{fb,ii}$ will increase $\hat{\kappa}^2_{i,i}$ in the signal term and also increase $\hat{\sigma}^2_{i,\bar{i}}$ in the interference term, and feedback power allocation will balance these different effects.

Denote $x\triangleq\frac{T_{fb}}{N_B}P_{fb,ii}L_{i,i}$, $a\triangleq\overline{T}_tP_tL_{i,\bar{i}}$, $b\triangleq P_dL_{i,\bar{i}}$, and $\rho\triangleq T_{fb}P^{ul}L_{i,i}$, then problem \eqref{eq:aFB_PA} is equivalent to
\begin{equation}
\max_{0\leq x\leq\rho}\frac{\frac{x}{1+x}}{1+b\cdot\frac{1+a+\rho-x}{(1+a)(1+\rho-x)}}.
\end{equation}
Denote $\lambda_1\triangleq1+\rho$, $\lambda_2\triangleq\frac{ab}{1+a+b}$, the objective function can be rewritten as
\begin{equation}
\frac{1+a}{1+a+b}\left\{1+\frac{1}{1+\lambda_1+\lambda_2}\left[\frac{(\lambda_1+\lambda_2)\lambda_2}{x-(\lambda_1+\lambda_2)}-\frac{1+\lambda_1}{x+1}\right]\right\}.
\end{equation}
So the maximization problem is equivalent to $\max_{0\leq x\leq\rho}f(x)$ with $f(x)\triangleq\frac{(\lambda_1+\lambda_2)\lambda_2}{x-(\lambda_1+\lambda_2)}-\frac{1+\lambda_1}{x+1}$.
As $\lambda_1>0$, $\lambda_2>0$, and $\lambda_1=\rho+1>x$, the first and second terms are both concave, so the objective function is concave.
Setting $\frac{\partial f(x)}{\partial x}=0$, we have
\begin{equation}\label{eq:lambda}
\left[(1+\lambda_1)-(\lambda_1+\lambda_2)\lambda_2\right]x^2-2(\lambda_1+\lambda_2)(1+\lambda_1+\lambda_2)x
+\lambda_1(\lambda_1+\lambda_2)(1+\lambda_1+\lambda_2)=0.
\end{equation}
Denote $(x_1^\star,x_2^\star)$ as the solution pair of \eqref{eq:lambda}, if $x_i^\star\in[0,\rho]$, $i=1,2$, then it is the solution for the original problem; otherwise, the maximal value is obtained at the edge and $x^\star=\rho$ is the solution, as $x=0$ makes the objective function to be 0 which is obviously not the maximum.

\subsection{Training with Digital Feedback}
In this part, we consider \emph{digital feedback}, also called \emph{finite-rate} or \emph{limited feedback} \cite{Jin06IT,LovHea08JSAC}, which feeds back quantized CSI. We assume user $i$ ($i=1,2$) feeds back a total of $B_i$ bits, among which $B_{i1}$ bits is for the channel estimate $\tilde{\bh}_{i,1}$ of BS 1 and $B_{i2}$ bits for the channel estimate $\tilde{\bh}_{i,2}$ of BS 2. The feedback channel is assumed to be error-free and without delay. The feedback interval is $T_{fb}=\mu B_i$, where $\mu$ is a conversion factor that relates bits to symbols, e.g., $\mu=1$ for BPSK feedback.

The channel estimate $\tilde{\bh}_{i,j}$, $i,j=1,2$, is fed back using a quantization codebook known at both the transmitter and receiver, which consists of unit norm vectors of size $2^{B_{ij}}$. We assume each user has multiple quantization codebooks, with the codebook of size $2^{B_{i,j}}$ denoted as $\mathcal{C}_{i,j}=\{\mathbf{c}_{1},\mathbf{c}_{2},\cdots,\mathbf{c}_{2^{B_{ij}}}\}$. Note that in practice the number of codebooks at each user may be limited to a small number, which is not considered in the paper. The quantized channel vector is $\hat{\bh}_{i,j}=\arg\max_{\bc_\ell\in\mathcal{C}_{i,j}}\left|{\frac{\tilde{\bh}_{i,j}}{\|\tilde{\bh}_{i,j}\|}}\bc_\ell\right|$.
The random vector quantization (RVQ) codebook \cite{Jin06IT,YueLov07Twc} is used to facilitate the analysis, where each quantization vector is independently chosen from the isotropic distribution on the $N_t$-dimensional unit sphere. Denote $\cos\theta_{i,j}=\left|\frac{\tilde{\bh}_{i,j}^*}{\|\tilde{\bh}_{i,j}\|}\hat{\bh}_{i,j}\right|$, and then with RVQ we have \cite{YueLov07Twc}
\begin{equation}\label{eq:xi}
\xi_{i,j}\triangleq\mathbb{E}_{\theta_{i,j}}\left[\cos^2\theta_{i,j}\right]=1-2^{B_{i,j}}\cdot\beta\left(2^{B_{i,j}},\frac{N_t}{N_t-1}\right)\geq1-2^{-\frac{B_{i,j}}{N_t-1}},
\end{equation}
where $\beta(x,y)$ is the Beta function.

We first derive the average achievable throughput for given $B_{i,j}$, $i,j=1,2$.

\begin{proposition}\label{Prop:dFB}
The average achievable throughput for user $i$ with training and digital feedback is approximated as
\begin{equation}\label{eq:Rate_dFB}
R_{i,\mathrm{dFB}}(s_1,s_2)\approx\left\{\begin{array}{ll}
\mathcal{R}^{(2)}({\kappa}_{i,i}^2\xi_{i,i}P_dL_{i,i},P_dL_{i,\bar{i}},N_t) & (s_i,s_{\bar{i}})=(\mathrm{BF},\mathrm{BF})\\
\mathcal{R}^{(3)}({\kappa}_{i,i}^2\xi_{i,i}P_dL_{i,i},{\sigma}_{i,\bar{i}}^2P_dL_{i,\bar{i}},\kappa_{i,\bar{i}}^22^{-\frac{B_{i,\bar{i}}}{N_t-1}}P_dL_{i,\bar{i}},N_t) & (s_i,s_{\bar{i}})=(\mathrm{BF},\mathrm{IC})\\
\mathcal{R}^{(3)}({\kappa}_{i,i}^2\xi_{i,i}P_dL_{i,i},{\sigma}_{i,\bar{i}}^2P_dL_{i,\bar{i}},\kappa_{i,\bar{i}}^22^{-\frac{B_{i,\bar{i}}}{N_t-1}}P_dL_{i,\bar{i}},N_t-1) & (s_i,s_{\bar{i}})=(\mathrm{IC},\mathrm{IC})\\
\mathcal{R}^{(2)}({\kappa}_{i,i}^2\xi_{i,i}P_dL_{i,i},P_dL_{i,\bar{i}},N_t-1) & (s_i,s_{\bar{i}})=(\mathrm{IC},\mathrm{BF})
\end{array}\right.
\end{equation}
where $\mathcal{R}^{(2)}$ and $\mathcal{R}^{(3)}$ are given by \eqref{eq:Rate_I2} and \eqref{eq:Rate_3cell}, respectively.
\end{proposition}
\begin{proof}
See Appendix \ref{App:dFB}.
\end{proof}

This result can be used for adaptive transmission strategy selection at low SNR. In the following, we optimize training and digital feedback for ICIC.

\subsubsection{Training Optimization}
Similar to \eqref{eq:Rbound_aFB}, we first get the following approximation for the rate loss due to training and digital feedback
\begin{equation}
R_i-R_{i,\mathrm{dFB}}\approx\log_2\left(1+L_{i,\bar{i}}P_d\left(\frac{1}{\overline{T}_tP_tL_{i,\bar{i}}}+2^{-\frac{B_{i,\bar{i}}}{N_t-1}}\right)\right),\quad i=1,2.
\end{equation}
With $P_d,P_t\rightarrow\infty$ and $\frac{P_d}{P_t}=\nu$, the rate loss is approximately $\log_2\left(1+\frac{\nu}{\overline{T}_t}+L_{i,\bar{i}}P_d2^{-\frac{B_{i,\bar{i}}}{N_t-1}}\right)$, which grows with $P_d$ for fixed $B_{i,\bar{i}}$. This shows that the system throughput is limited by the residual interference due to the quantization error. This effect can be eliminated if $B_{i,\bar{i}}$ is allowed to increase with the uplink transmission power.

Then we can get the following approximation for the average achievable rate for user $i$ ($i=1,2$)
\begin{equation}\label{eq:Rate_dFB2}
R_{i,\mathrm{dFB}}\approx\log_2\frac{L_{i,i}\xi_{i,i}e^{\psi(N_t-1)}}{P_d^{-1}+\left(\bar{T}_tP_t\right)^{-1}+
L_{i,\bar{i}}2^{-\frac{B_{i,\bar{i}}}{N_t-1}}}.
\end{equation}
This is similar to \eqref{eq:41} for analog feedback. Therefore, the optimal $(P^\star_d,P^\star_t)$ are also given in \eqref{eq:P_tOptaFB}, and they are independent of the feedback bit allocation in the uplink.

\subsubsection{Feedback Optimization}
We assume each user can adaptively select the number of feedback bits and apply the corresponding quantization codebook for channel feedback for different BSs. The feedback optimization is based on the following approximation for the average SINR
\begin{equation}
\overline{\mbox{SINR}}_i\approx\frac{P_dL_{i,i}\xi_{i,i}(N_t-1)}{1+P_dL_{i,\bar{i}}\sigma^2_{i,\bar{i}}+P_dL_{i,\bar{i}}\kappa_{i,\bar{i}}2^{-\frac{B_{i,\bar{i}}}{N_t-1}}},\quad i=1,2,
\end{equation}
which follows the distributions of signal and interference terms. Applying the bound in \eqref{eq:xi}, the feedback bit allocation problem is formulated as
\begin{equation}\label{eq:Bopt}
(B^\star_{i1},B^\star_{i2})=\arg\max_{B_{i1}+B_{i2}=B_i}\frac{1-2^{-\frac{B_{i,i}}{N_t-1}}}{1+P_dL_{i,\bar{i}}\sigma^2_{i,\bar{i}}+P_dL_{i,\bar{i}}\kappa^2_{i,\bar{i}}2^{-\frac{B_{i,\bar{i}}}{N_t-1}}},\quad i=1,2,
\end{equation}
which is done individually at each user. This is an integer programming problem. To get an analytical solution, we will first relax the constraint that $B_{i1,}$ and $B_{i2}$ have to be integers.

Denote $a_0\triangleq\frac{1+P_dL_{i,\bar{i}}\sigma^2_{i,\bar{i}}}{P_dL_{i,\bar{i}}\kappa^2_{i,\bar{i}}}$, $X_0\triangleq2^{-\frac{B_i}{N_t-1}}$, and $x\triangleq2^{-\frac{B_{i,i}}{N_t-1}}$. Then problem \eqref{eq:Bopt} is reformulated as
\begin{equation}
x^\star=\arg\max_{X_0<x<1}\frac{1-x}{\frac{X_0}{x}+a_0}.
\end{equation}
The objective function $f(x)\triangleq\frac{1-x}{\frac{X_0}{x}+a_0}$ is concave, and set $\frac{\partial f(x)}{\partial x}=0$, we have the solution $x^\star=\sqrt{\frac{X_0}{a_0}+\frac{X_0^2}{a_0^2}}-\frac{X_0}{a_0}$.
Then the following solution will be used for feedback bits allocation
\begin{align}
B_{i,i}^\star=\lfloor-(N_t-1)\log_2x^\star\rfloor,\quad
B_{i,\bar{i}}^\star=B_i-\lfloor-(N_t-1)\log_2x^\star\rfloor.\label{eq:dFB_BA1}
\end{align}


\section{Numerical Results}\label{Sec:Num}
In this section, we present numerical results to verify our analysis and demonstrate the performance gain with training and feedback optimization. We focus on the 2-cell model in Fig. \ref{fig:2cell}, where the users are located on the line connecting the two BSs. We assume that the downlink and uplink transmissions have the same power constraint, i.e., $P^{dl}=P^{ul}$. The cell radius is $R=1$ km, the pathloss exponent is $\alpha=3$, and $N_t=4$. As we do not consider training and feedback overhead optimization, the training and feedback intervals are fixed to be $T_t=N_BN_t$ and $T_{fb}=N_B^2N_t$, for $N_B=2$, respectively.

\subsection{Simulation vs. Approximation}
In Fig. \ref{fig:SimvsCal}, we compare the derived approximations and simulations for the achievable throughput with training and digital feedback. Both approximation \eqref{eq:Rate_dFB} and high-SNR approximation \eqref{eq:Rate_dFB2} are shown. For feedback bits $B$, we consider two scenarios: \textbf{Fixed feedback bits}: $\mu B=T_{fb}$, in the simulation, we use $\mu=1$, which represents BPSK feedback, as in \cite{SanHon10IT}; \textbf{Varying feedback bits}: the total number of feedback bits for user $i$ is $B_i=\lfloor T_{fb}\log_2(1+P^{ul}L_{i,i})\rfloor$, $i=1,2$, i.e., we assume perfect feedback at the rate of the uplink capacity, as in \cite{KobJin09}.

The accuracy of \eqref{eq:Rate_dFB} is shown in the fixed feedback case, where the high-SNR approximation \eqref{eq:Rate_dFB2} appears to be a lower bound. With varying feedback bits, as simulation is of prohibitive complexity with large $B$, we only compare two approximations, and it shows the high-SNR approximation \eqref{eq:Rate_dFB2} is accurate for this case. With varying feedback bits, the performance of ICIC with training and feedback approaches the perfect CSI case. However, in practice the number of feedback bits is normally fixed, and this will be the assumption in the following discussion.

\subsection{Training and Feedback Optimization}
In Fig. \ref{fig:OPTfig}, we compare the following different systems:
\begin{itemize}
\item Training + aFB I: analog feedback without optimization;
\item Training + aFB II: analog feedback with training optimization \eqref{eq:P_tOptaFB};
\item Training + aFB III: analog feedback with training and feedback optimization \eqref{eq:P_tOptaFB}, \eqref{eq:aFB_PA};
\item Training + dFB I: digital feedback without optimization;
\item Training + dFB II: digital feedback with training optimization \eqref{eq:P_tOptaFB};
\item Training + dFB III: digital feedback with training and feedback optimization \eqref{eq:P_tOptaFB}, \eqref{eq:dFB_BA1}.
\end{itemize}
We see that training and feedback optimization provides significant performance gains.
\begin{itemize}
\item With analog feedback, training optimization is more important, and additional uplink feedback power allocation provides limited performance gain. This is because the training power allocation is performed within $T-T_{fb}$ symbols, and $P_t\sim\sqrt{T}P^{dl}$ provides significant performance gain over equal power allocation; while the uplink feedback power allocation is within $T_{fb}$ symbols, and normally $T\gg T_{fb}$.
\item With digital feedback, training optimization alone provides little performance gain, and the feedback bit allocation is more important. This is because we assume a fixed number of feedback bits and the uplink is the limiting factor for the channel estimation accuracy, so the feedback optimization is more important.
\end{itemize}
These observations indicate that we can focus on the downlink training optimization for analog feedback while the uplink feedback should be carefully designed when employing digital feedback.

\subsection{Average and Edge Throughput}
In Fig. \ref{fig:R_PA}, the average sum throughput and edge throughput, which is represented by the 5th percentile throughput, are compared for different systems with 2 users randomly located on the line connecting 2 BSs. The adaptive ICIC system and the conventional single-cell beamforming system with perfect CSI are also shown for comparison. The system with and without training and feedback optimization are shown for both analog and digital feedbacks. Although these optimization procedures are developed for the high SNR regime, they will be used for all SNR values, as we have already observed the accuracy of the high-SNR approximation in Fig. \ref{fig:SimvsCal} and we are more interested in the high-SNR regime where multicell processing provides higher performance gain. Adaptive transmission strategy selection is applied once the training and feedback are optimized. We can see that the analog feedback system provides performance close to the perfect CSI case, and the digital feedback system with $B=T_{fb}$ is not as good but still provides significant gain over single-cell beamforming with perfect CSI. In addition, training and feedback optimization improves the performance of analog and digital feedback.

\section{Conclusions}\label{Sec:Conclusions}
This paper investigated an intercell interference cancellation (ICIC) system while specifically accounting for and optimizing the necessary channel training and feedback. A main conclusion is that despite recent pessimism, ICIC with practical CSI assumptions can efficiently combat intercell interference with adaptive strategies: at low SNR adaptive transmission strategy switching is required, while at high SNR training and feedback optimization can greatly improve the performance. Considering that ICIC poses light requirements on backhaul capacity, it appears to be a practical solution for interference mitigation in multicell networks. Although our analysis was limited to two cells to retain tractability, the results in the paper provide guidelines and a starting point for designing multicell processing systems, with many users per cell, including possibly cells that are highly heterogeneous in terms of coverage areas and layout \cite{AndBac10,DhiGan11}.

There are further key issues that require more investigation. For example, the uplink feedback channel is assumed to be an unfaded AWGN channel for analog feedback and as an error-free channel without delay for digital feedback. More realistic feedback channels should be considered in the future, such as the effect of fading, feedback error and delay.  In some cases, where the delay is a significant fraction of (or even exceeds) the channel coherence time, additional feedback rate or training optimization may prove irrelevant.  An alternative approach in that case might be to consider the recently proposed ``retrospective'' approaches to exploiting delayed CSI feedback \cite{MadTse10,MalJaf10}, which are robust to feedback delay and have a possibly large regime of utility \cite{XuAnd11}.

\useRomanappendicesfalse
\appendix

\subsection{Proof of Proposition \ref{Prop:Training}}\label{App:PropTraining}
With CSI training, the signal power of user $i$ ($i=1,2$) can be approximated as
\begin{align}
P_dL_{i,i}|\bh_{i,i}^*\bff_{i,s_i}|^2
\stackrel{(a)}{=}P_dL_{i,i}\left|(\tilde{\bh}_{i,i}+\bn_{i,i})^*\bff_{i,s_i}\right|^2
\stackrel{(b)}{\approx} P_dL_{i,i}\left|\tilde{\bh}_{i,i}^*\bff_{i,s_i}\right|^2
\stackrel{(c)}\sim\left\{\begin{array}{ll}P_dL_{i,i}\kappa_{i,i}^2\chi_{2N_t}^2&s_i=\mathrm{BF}\\
P_dL_{i,i}\kappa_{i,i}^2\chi_{2(N_t-1)}^2&s_i=\mathrm{IC}\end{array},\right.\notag
\end{align}
where step (a) follows \eqref{eq:H-Hest}, in (b) we ignore the estimation error term which is normally very small as its variance is inverse to the pilot power, and step (c) follows the results for the perfect CSI case and $\tilde{\bh}_{i,i}\sim\mathcal{CN}(0,\kappa_{i,i}\bI_{N_t})$. The interference term for $s_{\bar{i}}=\mathrm{IC}$ is
\begin{align}
P_dL_{i,\bar{i}}\left|\bh_{i,\bar{i}}^*\bff_{\bar{i},s_{\bar{i}}}\right|^2
=P_dL_{i,\bar{i}}\left|(\tilde{\bh}_{i,\bar{i}}+\bn_{i,\bar{i}})^*\bff_{\bar{i},s_{\bar{I}}}\right|^2
\stackrel{(d)}{=}P_dL_{i,\bar{i}}|\bn_{i,\bar{i}}^*\bff_{\bar{i},s_{\bar{i}}}|^2
\sim P_dL_{i,\bar{i}}\sigma_{i,\bar{i}}^2\chi_2^2,\notag
\end{align}
where step (d) follows the design rule of the precoder. If $s_{\bar{i}}=\mathrm{BF}$, then $P_dL_{i,\bar{i}}\left|\bh_{i,\bar{i}}^*\bff_{\bar{i},s_{\bar{i}}}\right|^2\sim P_dL_{i,\bar{i}}\chi_2^2$. Therefore, both the signal power and interference power are chi-square random variables. Applying  \emph{Lemma \ref{lemma:Rate}}, the average achievable throughput of user $i$ with estimated CSI is approximated as in \eqref{eq:Rate_CHest}.

\subsection{Proof of Lemma \ref{lemma:rateloss}}\label{App:rateloss}
The average throughput of user $i$ is approximated by
$R_{i,\mathrm{T}}\approx\mathbb{E}\left[\log_2\left(1+\frac{P_dL_{i,i}\kappa_{i,i}^2\chi_{2(N_t-1)}^2}{1+P_dL_{i,\bar{i}}\sigma_{i,\bar{i}}^2\chi_2^2}\right)\right]$.
With $P_d,P_t\rightarrow\infty$ and $\frac{P_d}{P_t}=\nu$, we have $\kappa_{i,i}^2\rightarrow1$ and
$P_d\sigma_{i,\bar{i}}^2\rightarrow\frac{\nu}{\overline{T}_tL_{i,\bar{i}}}$, so
\begin{align}
R_{i,\mathrm{T}}\rightarrow\mathbb{E}\left[1+\frac{P_dL_{i,i}\chi_{2(N_t-1)}^2}{1+\frac{\nu}{\overline{T}_t}\chi_2^2}\right]
\approx\mathbb{E}\left[\log_2\left(\frac{P_dL_{i,i}\chi_{2(N_t-1)}^2}{1+\frac{\nu}{\overline{T}_t}\chi_2^2}\right)\right]
\approx R_{i}-\mathbb{E}\left[\log_2\left(1+\frac{\nu}{\overline{T}_t}\chi_2^2\right)\right],
\end{align}
which gives the desired result.

\subsection{Proof of Lemma \ref{lemma:SuffCond}}\label{App:SuffCond}
Denote $x\triangleq\sqrt{1-{\overline{T}_t}/{\overline{T}}}$, $a\triangleq\log(L_{ii}P^{dl}e^{\psi(N_t-1)})$, the optimization problem can be rewritten as
\begin{equation}
\max_{0<x\leq \sqrt{1-{N_B}/{\overline{T}}}}f(x)=ax^2-2x^2\log\left(x+1/\sqrt{\overline{T}}\right).
\end{equation}
To obtain the maximum at $\overline{T}_t=N_B$, or $x=\sqrt{1-{N_B}/{\overline{T}}}$, a sufficient condition is $\frac{\partial f}{\partial x}>0$ for $0<x\leq \sqrt{1-{N_B}/{\overline{T}}}$. First, we have
$\frac{\partial f}{\partial x}=2ax-\frac{2x^2}{x+1/\sqrt{\overline{T}}}-4x\log\left(x+1/\sqrt{\overline{T}}\right)$.
Set it to be greater than 0, we get
\begin{equation}\label{eq:abound}
a>\frac{x}{x+1/\sqrt{\overline{T}}}+2\log\left(x+1/\sqrt{\overline{T}}\right)\triangleq g(x).
\end{equation}
As the right-hand side is an increasing function of $x$, and $0<x\leq \sqrt{1-\frac{N_B}{\overline{T}}}$, a sufficient condition for \eqref{eq:abound} to hold is $a>g\left(\sqrt{1-{N_B}/{\overline{T}}}\right)$, which gives the result in \eqref{eq:SuffCond}.

\subsection{Proof of Proposition \ref{Prop:dFB}}\label{App:dFB}
With training and digital feedback, the precoders are designed based on the quantized CSI $\hat{\bh}_{i,j}$, $i,j=1,2$. The signal power for user $i$ is
\begin{align}
&P_dL_{i,i}|\bh_{i,i}^*\bff_{i,s_i}|^2
=P_dL_{i,i}|(\tilde{\bh}_{i,i}+\bee_{i,i})^*\bff_{i,s_i}|^2\notag\\
\stackrel{(a)}{\approx}& P_dL_{i,i}|\tilde{\bh}_{i,i}^*\bff_{i,s_i}|^2
=P_dL_{i,i}\|\tilde{\bh}_{i,i}\|\cdot\left|\frac{\tilde{\bh}_{i,i}^*}{\|\tilde{\bh}_{i,i}\|}\cdot\bff_{i,s_i}\right|^2\notag\\
\stackrel{(b)}{=}&P_dL_{i,i}\|\tilde{\bh}_{i,i}\|\cdot\left|(\cos\theta_{i,i}\hat{\bh}_{i,i}+\sin\theta_{i,i}\bg_{i,i})^*\cdot\bff_{i,s_i}\right|^2\notag\\
\stackrel{(c)}{\approx}&P_dL_{i,i}\xi_{i,i}\|\tilde{\bh}_{i,i}\|\cdot|\hat{\bh}_{i,i}^*\bff_{i,s_i}|^2
\stackrel{(d)}\sim\left\{\begin{array}{ll}P_dL_{i,i}\xi_{i,i}\kappa_{i,i}^2\chi_{2N_t}^2&s_i=\mathrm{BF}\\
P_dL_{i,i}\xi_{i,i}\kappa_{i,i}^2\chi_{2(N_t-1)}^2&s_i=\mathrm{IC}\end{array},\right.\notag
\end{align}
where step (a) ignores the estimation error vector which is small for large $P_t$, in step (b) the estimated channel direction is decomposed into the quantized direction $\hat{\bh}_{i,i}$ and its orthogonal direction $\bg_{i,i}$, step (c) ignores the quantization error term, as $\sin\theta_{i,i}=\mathcal{O}(\theta_{i,i})\ll1$ for reasonable $B_{i,i}$, and approximates $\cos\theta_{i,i}^2$ by its expectation, and step (d) follows the results for the perfect CSI case and the fact $\tilde{\bh}_{i,i}\sim\mathcal{CN}(0,\kappa_{i,i}\bI_{N_t})$.

If $s_{\bar{i}}=\mathrm{BF}$, the interference power is $P_dL_{i,\bar{i}}\left|\bh_{i,\bar{i}}^*\bff_{\bar{i},s_{\bar{i}}}\right|^2\sim P_dL_{i,\bar{i}}\chi_2^2$; otherwise, for $s_{\bar{i}}=\mathrm{IC}$,
\begin{align}
&P_dL_{i,\bar{i}}|\bh_{i,\bar{i}}^*\bff_{\bar{i},s_{\bar{i}}}|^2
=P_dL_{i,\bar{i}}\left|(\tilde{\bh}_{i,\bar{i}}+\bee_{i,\bar{i}})^*\bff_{\bar{i},s_{\bar{i}}}\right|^2\notag\\
\stackrel{(e)}{\approx}&P_dL_{i,\bar{i}}\left|\tilde{\bh}_{i,\bar{i}}^*\bff_{\bar{i},s_{\bar{i}}}\right|^2+P_dL_{i,\bar{i}}\left|\bee_{i,\bar{i}}^*\bff_{\bar{i},s_{\bar{i}}}\right|^2
\stackrel{(f)}{\sim}P_dL_{i,\bar{i}}\kappa_{i,\bar{i}}^2\cdot2^{-\frac{B_{i,\bar{i}}}{N_t-1}}\chi_2^2+P_dL_{i,\bar{i}}\sigma_{i,\bar{i}}^2\chi_2^2,\label{eq:ResidualI}
\end{align}
where step (e) ignores the term with both $\tilde{\bh}_{i,\bar{i}}$ and $\bee_{i,\bar{i}}$, and step (f) follows the result in \cite{ZhaRob09EURASIP} that the residual interference term due to quantization error, $\left|\tilde{\bh}_{i,\bar{i}}^*\bff_{\bar{i},s_{\bar{i}}}\right|^2$, is an exponential random variable. The two terms in \eqref{eq:ResidualI} correspond to the residual interference due to the quantization error and the estimation error, respectively.

Based on the above results and applying \emph{Lemma \ref{lemma:Rate}}, we get \eqref{eq:Rate_dFB}.

\bibliographystyle{IEEEtran}
\bibliography{bibi}

\begin{thebibliography}{10}
\providecommand{\url}[1]{#1}
\csname url@rmstyle\endcsname
\providecommand{\newblock}{\relax}
\providecommand{\bibinfo}[2]{#2}
\providecommand\BIBentrySTDinterwordspacing{\spaceskip=0pt\relax}
\providecommand\BIBentryALTinterwordstretchfactor{4}
\providecommand\BIBentryALTinterwordspacing{\spaceskip=\fontdimen2\font plus
\BIBentryALTinterwordstretchfactor\fontdimen3\font minus
  \fontdimen4\font\relax}
\providecommand\BIBforeignlanguage[2]{{%
\expandafter\ifx\csname l@#1\endcsname\relax
\typeout{** WARNING: IEEEtran.bst: No hyphenation pattern has been}%
\typeout{** loaded for the language `#1'. Using the pattern for}%
\typeout{** the default language instead.}%
\else
\language=\csname l@#1\endcsname
\fi
#2}}

\bibitem{GesHan10JSAC}
D.~Gesbert, S.~Hanly, H.~Huang, S.~Shamai, O.~Simeone, and W.~Yu, ``Multi-cell
  {MIMO} cooperative networks: A new look at interference,'' \emph{IEEE J.
  Select. Areas Commun.}, vol.~28, no.~9, pp. 1380--1408, Dec. 2010.

\bibitem{IrmDro11}
R.~Irmer, H.~Droste, P.~Marsch, M.~Grieger, G.~Fettweis, S.~Brueck, H.-P.
  Mayer, L.~Thiele, and V.~Jungnickel, ``Coordinated multipoint: Concepts,
  performance, and field trial results,'' \emph{IEEE Commun. Mag.}, vol.~49,
  pp. 102--111, Feb. 2011.

\bibitem{AnnBar10}
S.~Annapureddy, A.~Barbieri, S.~Geirhofer, S.~Mallik, and A.~Gorokhov,
  ``Coordinated joint transmission in {WWAN},'' in \emph{IEEE Communication
  Theory Workshop}, May 2010.

\bibitem{ZhaAnd10JSAC}
J.~Zhang and J.~G. Andrews, ``Adaptive spatial intercell interference
  cancellation in multicell wireless networks,'' \emph{IEEE J. Select. Areas
  Commun.}, vol.~28, no.~9, pp. 1455--1468, Dec. 2010.

\bibitem{ShamaiVTC01}
S.~{Shamai (Shitz)} and B.~M. Zaidel, ``Enhancing the cellular downlink
  capacity via co-processing at the transmitting end,'' in \emph{Proc. IEEE
  Veh. Technol. Conf.}, Rhodes, Greece, May 2001, pp. 1745--1749.

\bibitem{Zhang04}
H.~Zhang and H.~Dai, ``Cochannel interference mitigation and cooperative
  processing in downlink multicell multiuser {MIMO} networks,'' \emph{EURASIP
  Journal on Wireless Communications and Networking}, no.~2, pp. 222--235, 4th
  Quarter 2004.

\bibitem{Karakayali06a}
K.~Karakayali, G.~J. Foschini, and R.~A. Valenzuela, ``Network coordination for
  spectrally efficient communications in cellular systems,'' \emph{IEEE
  Wireless Communications Magazine}, vol.~13, no.~4, pp. 56--61, Aug. 2006.

\bibitem{ChoAnd08Twc}
W.~Choi and J.~G. Andrews, ``The capacity gain from intercell scheduling in
  multi-antenna systems,'' \emph{IEEE Trans. Wireless Commun.}, vol.~7, no.~2,
  pp. 714--725, Feb. 2008.

\bibitem{ZhaChe09Twc}
J.~Zhang, R.~Chen, J.~G. Andrews, A.~Ghosh, and R.~W. {Heath Jr.}, ``Networked
  {MIMO} with clustered linear precoding,'' \emph{IEEE Trans. Wireless
  Commun.}, vol.~8, no.~4, pp. 1910--1921, Apr. 2009.

\bibitem{MarFet07EW}
P.~Marsch and G.~Fettweis, ``A framework for optimizing the downlink
  performance of distributed antenna systems under a constrained backhaul,'' in
  \emph{Proc. European Wireless Conf. (EW' 07)}, Paris, France, Apr. 2007.

\bibitem{SanSom09IT}
A.~Sanderovich, O.~Somekh, H.~V. Poor, and S.~Shamai, ``Uplink macro diversity
  of limited backhaul cellular network,'' \emph{IEEE Trans. Inform. Theory},
  vol.~55, no.~8, pp. 3457--3478, Aug. 2009.

\bibitem{BjoZak09Glob}
E.~Bjornson, R.~Zakhour, D.~Gesbert, and B.~Ottersten, ``Distributed multicell
  and multiantenna precoding: characterization and performance evaluation,'' in
  \emph{Proc. IEEE Globecom}, Honolulu, Hawaii, Nov. 30 - Dec. 4 2009.

\bibitem{ZakGes10Twc}
R.~Zakhour and D.~Gesbert, ``Distributed multicell-{MISO} precoding using the
  layered virtual {SINR} framework,'' \emph{IEEE Trans. Wireless Commun.},
  vol.~9, no.~8, pp. 2444--2448, Aug. 2010.

\bibitem{RamCai09PIMRC}
S.~Ramprashad and G.~Caire, ``Cellular vs network {MIMO}: A comparison
  including channel state information overhead,'' in \emph{Proc. of the IEEE
  Int. Symp. on Personal Indoor and Mobile Radio Comm.}, Tokyo, Japan, Sept.
  2009.

\bibitem{RamCai09Asilomar}
S.~Ramprashad, G.~Caire, and H.~Papadopoulos, ``Cellular and network {MIMO}
  architectures: {MU-MIMO} spectral efficiency and costs of channel state
  information,'' in \emph{Proc. IEEE Asilomar Conference on Signals, Systems,
  and Computers}, Pacific Grove, CA, Nov. 2009.

\bibitem{HuhTul10CISS}
H.~Huh, A.~Tulinoy, and G.~Caire, ``Network {MIMO} large-system analysis and
  the impact of {CSIT} estimation,'' in \emph{44th Annual Conference on
  Information Sciences and Systems (CISS)}, Princeton, NJ, Mar. 2010, pp. 1--6.

\bibitem{BhaRao10ICASSP}
R.~Bhagavatula, B.~Rao, and R.~W. {Heath, Jr.}, ``Limited feedback with joint
  {CSI} quantization for multicell cooperative generalized eigenvector
  beamforming,'' in \emph{Proc. of the IEEE Int. Conf. on Acoustics, Speech,
  and Signal Proc.}, Dallas, TX, Mar. 2010, pp. 2838--2841.

\bibitem{BhaHea11Tsp}
R.~Bhagavatula and R.~W. {Heath, Jr.}, ``Adaptive limited feedback for sum-rate
  maximizing beamforming in cooperative multicell systems,'' \emph{IEEE Trans.
  Signal Processing}, vol.~59, no.~2, pp. 800--811, Feb. 2011.

\bibitem{LovHea08JSAC}
D.~J. Love, R.~W. {Heath Jr.}, V.~K.~N. Lau, D.~Gesbert, B.~D. Rao, and
  M.~Andrews, ``An overview of limited feedback in wireless communication
  systems,'' \emph{IEEE J. Select. Areas Commun.}, vol.~26, no.~8, pp.
  1341--1365, Oct. 2008.

\bibitem{Jin06IT}
N.~Jindal, ``{MIMO} broadcast channels with finite rate feedback,'' \emph{IEEE
  Trans. Inform. Theory}, vol.~52, no.~11, pp. 5045--5059, Nov. 2006.

\bibitem{ShaHas05IT}
M.~Sharif and B.~Hassibi, ``On the capacity of {MIMO} broadcast channels with
  partial side information,'' \emph{IEEE Trans. Inform. Theory}, vol.~51,
  no.~2, pp. 506--522, Feb. 2005.

\bibitem{YooJin07JSAC}
T.~Yoo, N.~Jindal, and A.~Goldsmith, ``Multi-antenna downlink channels with
  limited feedback and user selection,'' \emph{IEEE J. Select. Areas Commun.},
  vol.~25, no.~7, pp. 1478--1491, Sept. 2007.

\bibitem{ZhaRob09EURASIP}
J.~Zhang, R.~{W. Heath Jr.}, M.~Kountouris, and J.~G. Andrews, ``Mode switching
  for the multi-antenna broadcast channel based on delay and channel
  quantization,'' \emph{EURASIP Journal on Advances in Signal Processing},
  2009, article ID 802548, 15 pages.

\bibitem{ZhaKou11Tcomm}
J.~Zhang, M.~Kountouris, J.~G. Andrews, and R.~W. {Heath Jr.}, ``Multi-mode
  transmission for the {MIMO} broadcast channel with imperfect channel state
  information,'' \emph{IEEE Trans. Commun.}, vol.~59, no.~3, pp. 803--814, Mar.
  2011.

\bibitem{CaiJin10IT}
G.~Caire, N.~Jindal, M.~Kobayashi, and N.~Ravindran, ``Multiuser {MIMO}
  achievable rates with downlink training and channel state feedback,''
  \emph{IEEE Trans. Inform. Theory}, vol.~56, no.~6, pp. 2845--2866, Jun. 2010.

\bibitem{KobJin09}
M.~Kobayashi, N.~Jindal, and G.~Caire, ``Training and feedback optimization for
  multiuser {MIMO} downlink,'' \emph{IEEE Trans. Commun.}, 2011, to appear.

\bibitem{SanHon10IT}
W.~Santipach and M.~L. Honig, ``Optimization of training and feedback overhead
  for beamforming over block fading channels,'' \emph{IEEE Trans. Inform.
  Theory}, vol.~56, no.~12, pp. 6103--6115, Dec. 2010.

\bibitem{JinAnd10Tcomm}
N.~Jindal, J.~Andrews, and S.~Weber, ``Multi-antenna communication in ad hoc
  networks: Achieving {MIMO} gains with {SIMO} transmission,'' \emph{IEEE
  Trans. Commun.}, vol.~59, no.~2, pp. 529--540, Feb. 2011.

\bibitem{Kay93I}
S.~M. Kay, \emph{Fundamentals of statistical signal processing I: Estimation
  theory}.\hskip 1em plus 0.5em minus 0.4em\relax Prentice Hall, 1993.

\bibitem{KobCai08ISIT}
M.~Kobayashi, G.~Caire, and N.~Jindal, ``How much training and feedback are
  needed in {MIMO} broadcast channels?'' in \emph{Proc. IEEE Int. Symp.
  Information Theory}, Toronto, Canada, Jul. 2008, pp. 2663--2667.

\bibitem{HasHoc03IT}
B.~Hassibi and B.~Hochwald, ``How much training is needed in multiple-antenna
  wireless links?'' \emph{IEEE Trans. Inform. Theory}, vol.~49, no.~4, pp.
  951--963, Apr. 2003.

\bibitem{MarHoc06Tsp}
T.~L. Marzetta and B.~M. Hochwald, ``Fast transfer of channel state information
  in wireless systems,'' \emph{IEEE Trans. Signal Processing}, vol.~54, no.~4,
  pp. 1268--1278, Apr. 2006.

\bibitem{SamMan06Tcomm}
D.~Samardzija and N.~Mandayam, ``Unquantized and uncoded channel state
  information feedback in multiple-antenna multiuser systems,'' \emph{IEEE
  Trans. Commun.}, vol.~54, no.~7, pp. 1335--1345, Jul. 2006.

\bibitem{YueLov07Twc}
C.~K. Au-Yeung and D.~J. Love, ``On the performance of random vector
  quantization limited feedback beamforming in a {MISO} system,'' \emph{IEEE
  Trans. Wireless Commun.}, vol.~6, no.~2, pp. 458--462, Feb. 2007.

\bibitem{AndBac10}
J.~G. Andrews, F.~Baccelli, and R.~K. Ganti, ``A tractable approach to coverage
  and rate in cellular networks,'' submitted to \emph{IEEE Trans. on
  Communications}, Sept. 2010. Available at http://arxiv.org/abs/1009.0516.

\bibitem{DhiGan11}
H.~Dhillon, R.~K. Ganti, F.~Baccelli, and J.~G. Andrews, ``Modeling and
  analysis of {K}-tier downlink heterogeneous cellular networks,'' submitted to
  \emph{IEEE J. Select. Areas Commun.}, March 2011. Available at
  http://arxiv.org/abs/1103.2177.

\bibitem{MadTse10}
M.~A. Maddah-Ali and D.~Tse, ``Completely stale transmitter channel state
  information is still very useful,'' in \emph{Allerton Conference on
  Communication, Control and Computing}, Monticello, IL, Sept. 2010.

\bibitem{MalJaf10}
H.~Maleki, S.~A. Jafar, and S.~Shamai, ``Retrospective interference
  alignment,'' Sept. 2010, available at http://arxiv.org/abs/1009.3593.

\bibitem{XuAnd11}
J.~Xu, J.~G. Andrews, and S.~A. Jafar, ``Broadcast channels with delayed
  finite-rate feedback: Predict or observe?'' submitted to \emph{IEEE Trans. on
  Wireless Communications}, May 2011. Available at
  http://arxiv.org/abs/1105.3686.

\end{thebibliography}

\begin{figure}
\centering
\includegraphics[clip=true,scale=.8]{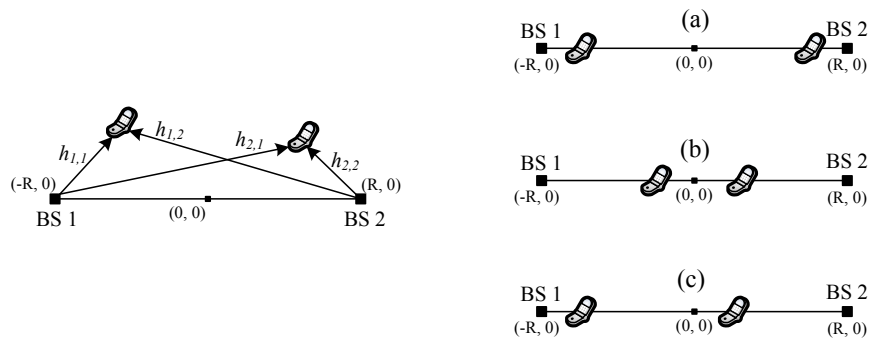}
\caption{A two-cell network. Each BS is serving a home user, which is suffering interference from the neighboring BS. We mainly consider the scenario where both users are on the line connecting BS 1 and BS 2, with three typical cases shown in the right figure.}\label{fig:2cell}
\end{figure}


\begin{figure*}
\centering {\subfigure[Perfect CSI]{\includegraphics[width=3in]{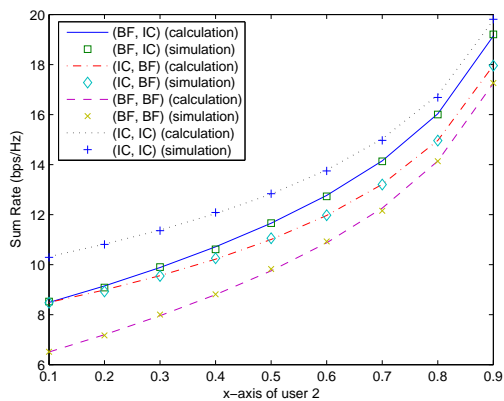} \label{fig:simvscal_10dB}}
\hfil \subfigure[CSI Training]{\includegraphics[width=3in]{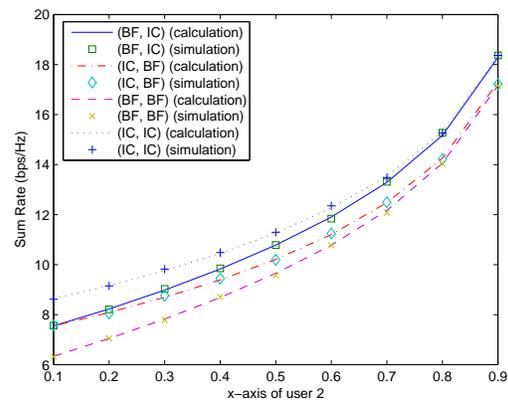}
\label{fig:Training10dB}}} \caption{Simulation and calculation results for the 2-cell network with perfect CSI and CSI training. $\alpha=3$, $N_t=4$. User 1 is at $(-.1R,0)$, edge SNR $10$ dB, $\alpha=3$,
$N_t=4$.} \label{fig:User1Edge}
\end{figure*}

\begin{figure*}
\centering {\subfigure[Perfect CSI]{\includegraphics[width=3in]{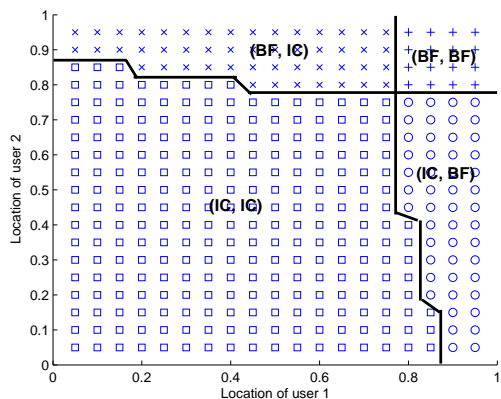} \label{fig:ModePlot5dBCSIT}}
\hfil \subfigure[CSI Training]{\includegraphics[width=3in]{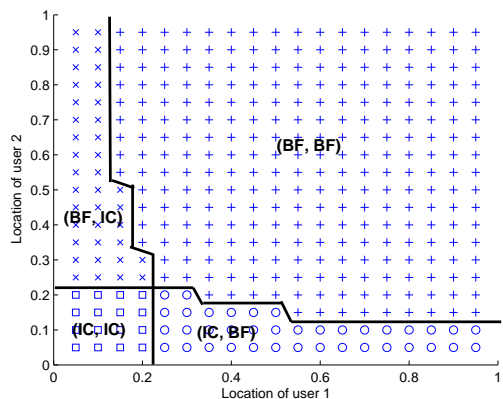}
\label{fig:ModePlot5dB}}} \caption{Operating regions for different transmission strategy pairs. $\alpha=3$, $N_t=4$, edge SNR 4 dB, user 1 and
user 2 are on the line connecting BS 1 and BS 2, `x':
$(s_1,s_2)=(\mathrm{BF},\mathrm{IC})$; 'o': $(s_1,s_2)=(\mathrm{IC},\mathrm{BF})$; '+':
$(s_1,s_2)=(\mathrm{BF},\mathrm{BF})$; '$\Box$': $(s_1,s_2)=(\mathrm{IC},\mathrm{IC})$.} \label{fig:ModePlot}
\end{figure*}


\begin{figure}
\centering
\includegraphics[width=4.5in]{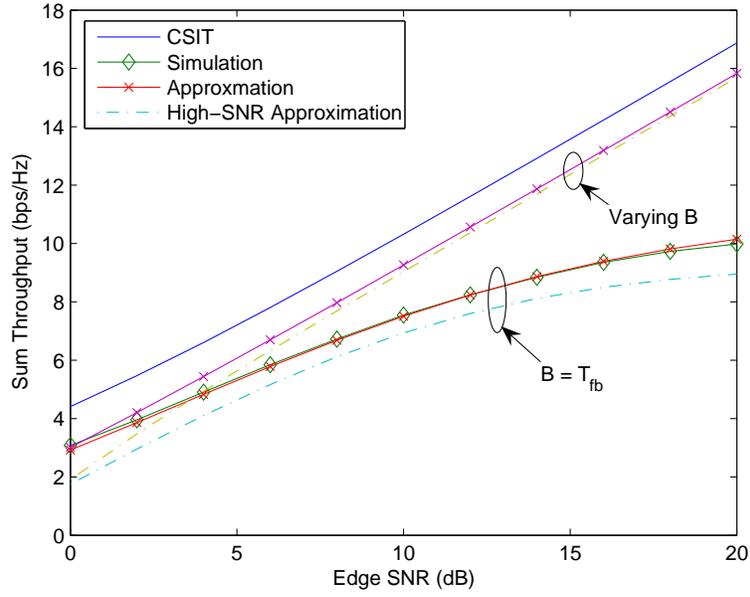}
\caption{Simulation and approximation for ICIC with training and digital feedback, $T=500$. User 1 is at $(-0.1R,0)$, and user 2 is at $(0.1R,0)$.}\label{fig:SimvsCal}
\end{figure}

\begin{figure}
\centering
\includegraphics[width=4.5in]{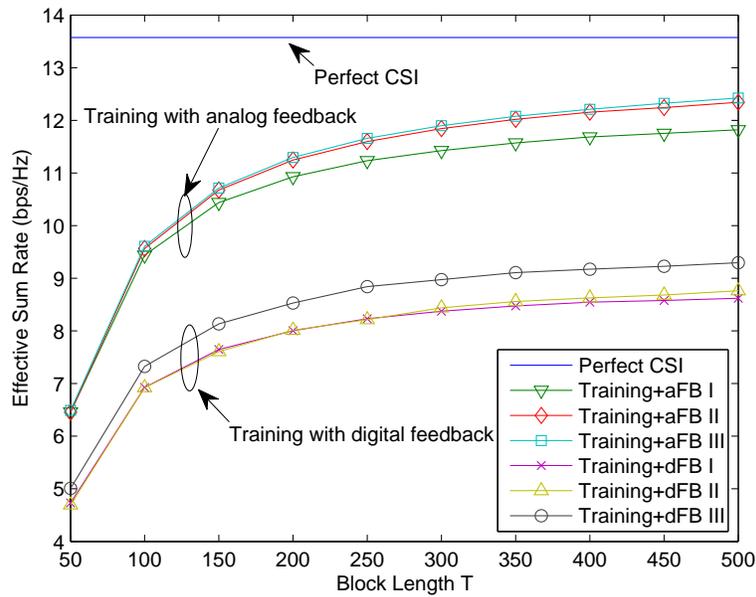}
\caption{Achievable sum rates for different systems, with edge SNR 15 dB, and $B=T_{fb}$ for digital feedback. User 1 is at $(-0.1R,0)$ and user 2 is at $(0.1R,0)$.}\label{fig:OPTfig}
\end{figure}

\begin{figure*}
\centering {\subfigure[Average Throughput]{\includegraphics[width=4.5in]{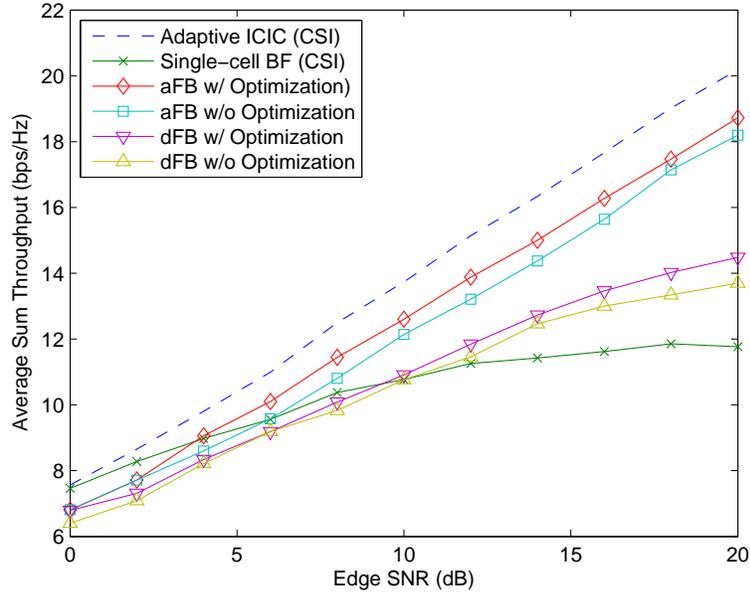} \label{fig:Ravg_PA}}
\hfil \subfigure[Edge Throughput]{\includegraphics[width=4.5in]{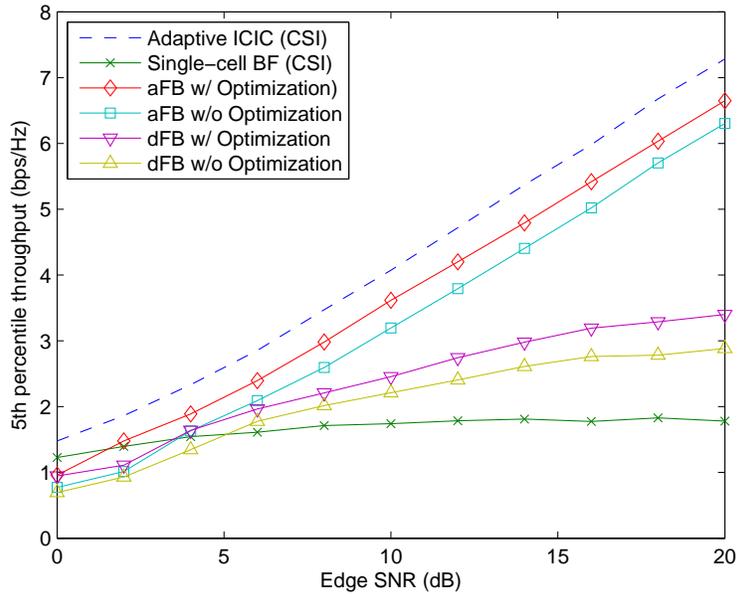}
\label{fig:R5_PA}}} \caption{Performance comparison of different systems, with $T=500$, and $B=T_{fb}$ for  digital feedback. ``aFB'' denotes the system with training and analog feedback, while ``dFB'' denotes the system with training and digital feedback. Training and feedback overheads are considered.} \label{fig:R_PA}
\end{figure*}

\end{document}